\documentclass[11pt]{article}
\usepackage[left=1in,right=1in,top=1in,bottom=1in]{geometry}
\usepackage{amsmath, amsfonts,amssymb, amsthm}
\usepackage{natbib}
\usepackage{graphicx,graphics}
\usepackage{multirow}
\usepackage{float}
\usepackage{authblk}
\usepackage{hyperref}
\usepackage{subfig}
\usepackage{dsfont}
\usepackage{enumerate}
\usepackage{color}
\usepackage{xcolor}

\usepackage{xcolor}

\theoremstyle{plain}

\newtheorem{theorem}{Theorem}[section]
\newtheorem*{thm*}{Theorem}

\newtheorem{corollary}[theorem]{Corollary}
\newtheorem{lemma}[theorem]{Lemma}

\theoremstyle{definition}

\theoremstyle{remark}
\newtheorem{remark}[theorem]{Remark}
\newtheorem{example}[theorem]{Example}

\newcommand\range{\operatorname{Range}}

\newcommand{\N}{\mathcal{N}}

\newcommand{\calA}{\mathcal{A}}

\def\hpi{\hat{\pi}}

\def\R{{\mathds R}}

\newcommand{\Exp}{\mathds{E}}

\newcommand{\Prob}{\mathds{P}}

\newcommand{\abs}[1]{\left\vert#1\right\vert}
\newcommand{\Set}[1]{\left\{#1\right\}}
\newcommand{\seq}[1]{\left<#1\right>}

\begin{document}

\title{Optimal investment with  insurable background risk and nonlinear portfolio allocation frictions}


\author[1]{Hugo E. Ramirez \thanks{\href{mailto:hugoedu.ramirez@urosario.edu.co}{hugoedu.ramirez@urosario.edu.co}}}

\author[1]{Rafael Serrano \thanks{\href{mailto:rafael.serrano@urosario.edu.co}{rafael.serrano@urosario.edu.co}}}

\affil[1]{Universidad del Rosario, Calle 12C No. 4-69, Bogot\'a, Colombia}

\maketitle

\abstract{
We study investment and insurance demand decisions for an agent in a theoretical continuous-time expected utility maximization model that combines risky assets with an (exogenous) insurable background risk. This risk takes the form of a jump-diffusion process with negative jumps in the return rate of the (self-financed) wealth.  The main distinctive feature of our model is that the agent's decision on portfolio choice and insurance demand causes nonlinear friction in the dynamics of the wealth process. We use the dynamic programming approach to find optimality conditions under which the agent assumes the insurable risk entirely, or partially, or purchases total insurance against it. In particular, we consider differential and piece-wise linear portfolio allocation frictions, with differential borrowing and lending rates as our most emblematic example. Finally, we present a mutual-fund separation result and illustrate our results with several numerical examples when the adverse jump risk has Beta distribution.
}


\section{Introduction}

Economic and financial decision-making under uncertainty usually occurs in the presence of multiple risks and incomplete markets. Portfolio choices and decisions
about endogenous risks sometimes must be made while simultaneously facing one or
more  exogenous ``background risks". Background risk typically refers to the uncertainty that affects
a decision maker's wealth level but cannot be hedged in the financial
markets. Most early published works on background risk are set in a single-period framework and employ the notions of stochastic dominance or risk vulnerability. \cite{eeckhoudt1996changes} examine background wealth deteriorations taking the form of changes in risk, in both general first- and
second-degree stochastic dominance, and find necessary and sufficient conditions for each of
these two types changes in background risk to imply a more risk-averse behavior. \cite{gollier1996risk} introduce the concept of risk vulnerability: additional and independent background risk increases a decision maker’s risk aversion when that risk is unfair (i.e., has non-positive expected value).

In the context of investment decisions, most works on background risk assume it is independent of portfolio choice decisions. However, for
many decision-making situations, some dependence between investment and background risks is expected. \cite{heaton2000portfolio} present a decision-theoretic model for portfolio selection with background risk and find considerable correlations
of stock returns with wage income and proprietary
income. Moreover, empirical
evidence on exchange-rate exposure of U.S. multinational
firms also demonstrate that correlations between returns
and exchange rates may be positive or negative,
depending on the company’s profile (e.g., exporting
operations versus importing operations), see e.g., \cite{allayannis2001exchange}.

In the related literature, there are a few exceptions to the typical assumption
of independence. For instance, \cite{tsetlin2005risky} study optimal decisions in a one-period expected utility framework in which the decision maker undertakes risky projects that are correlated with both additive
and multiplicative background risks. They document the
importance of the direction and degree of dependence between
the project risk and the background risk, and prove that it may
be optimal for a risk-averse agent to undertake a project with zero
or even negative expected returns in the presence of an additive
negatively correlated background risk.  \cite{franke2006multiplicative} study a similar problem for multiplicative background risk, extending the results of \cite{gollier1996risk} and finding conditions on
preferences that lead to more cautious behavior. \cite{franke2011risk} examine simultaneous effects of both additive and multiplicative risks on optimal portfolio choice. They  rationalize certain paradoxical choice behavior in that context and show how background risks might lead to a seemingly U-shaped relative risk aversion for a representative investor.

Although background risk usually emerges from many non-financial events such as pure insurance losses, inflation, labor
income, political turmoil, natural disasters, real estate losses, tax liabilities, etc., it may sometimes be hedged away partially or totally via insurance. In this case, inter-temporal models are more appropriate when insurance demand for background risk is analyzed jointly with portfolio or consumption decisions. Early works on inter-temporal insurance demand include \cite{briys1986insurance} who considers a continuous-time model,  assumes that the loss is proportional to the wealth and
alludes to an optimal coinsurance, proving that if the utility function is isoelastic then optimal insurance has the form of coinsurance and the deductible must be a function of wealth (and time).

In a similar model, \cite{gollier1994insurance} analyzes the optimal dynamic strategy of a risk-averse agent bearing an insurable risk to determine whether precautionary saving is superior to insurance in the long run. He assumes that the loss risk follows a Poisson process and that the loss function is not a function of wealth. \cite{gollier1994insurance} also shows that the demand for insurance vanishes in the long run if the loading factor exceeds a given strictly positive critical value, derives the optimal strategy for capital accumulation and insurance demand in the constant relative risk aversion case, and proves that compulsory full insurance reduces the rate of consumption if and only if risk aversion exceeds one.  \cite{touzi2000optimal} studies the stochastic control problem of maximizing expected utility from terminal wealth where the wealth process  is subject to shocks
produced by a general marked point process, and the agent allocates wealth between a
nonrisky asset and a (costly) insurance strategy which allows reducing
the level of the shocks, similar to the setting of \cite{briys1986insurance} and \cite{gollier1994insurance}. \cite{touzi2000optimal} connects the agent’s optimization problem to a suitable dual stochastic control problem in which the constraint on
the insurance strategy disappears,  finds a general existence result
for the dual problem, and obtains an explicit characterization for power (and logarithmic)
utility functions and linear insurance premium.

\cite{moore2006optimal} use the dynamic programming (DP) approach to extend \cite{briys1986insurance} and \cite{gollier1994insurance} by allowing the horizon to be random, and find  that if the premium is proportional to the expected payout, then the optimal per-claim insurance is deductible insurance. \cite{perera2010optimal} uses the convex duality  approach to obtain closed-form solutions for the optimal investment, consumption and insurance strategies of an individual with CARA utility and in the presence of an insurable risk that follows a Lévy process. \cite{lin2012risky} examine risky asset allocation and consumption rate decisions in the presence of background risk, using the DP approach. They find that the optimal allocation in risky assets reflects the agent’s risk attitude when background risk is related to investment risk, and that optimal insurance hedges investment risk and balances the growth and the volatility of consumption. \cite{mnif2012numerical,mnif2013optimal} studies the problem of finding the
optimal insurance strategy that reduces exposure to jump risk, similar to the setting of \cite{touzi2000optimal}, using  a suitable dual stochastic control problem. \cite{mnif2012numerical} characterizes the dual value function as the unique
viscosity solution of the corresponding Hamilton-Jacobi-Bellman (HJB) variational inequality, while \cite{mnif2013optimal}  approximates the optimal solution numerically  using a policy-iteration algorithm. 


The present paper studies joint decisions regarding risky asset allocation and insurance demand for a representative risk-averse agent in a finite-horizon continuous-time model, comparable with \cite{touzi2000optimal} and \cite{lin2012risky}. Taking dynamic portfolio selection in a Black-Scholes-type model as a generic example for decisions under risk, we extend the standard model of combining risky assets with an (exogenous) insurable background risk in the form of a jump-diffusion process  with negative jumps. The main distinctive feature of our model is that the agent's decision on portfolio choice and insurance cause nonlinear frictions in the dynamics of the wealth process. More concretely, we consider in the agent's wealth mean rate of return a (possibly non-linear) drift friction term $f(\pi,\kappa)$ that depends on the portfolio allocation weights $\pi$ and the fraction $\kappa$ of insurable risk assumed by the agent\footnote{$1-\kappa$ is the agent's insurance demand.}. This friction term incorporates the rate of return of an additional endogenous cash flow relative to the value of investments in the financial market and insurance costs. We will focus on two cases: first, $f$ separable having the form $f(\pi,\kappa)=g(\pi)-p(\kappa)$ with $p$ differentiable and $g$ either smooth or piece-wise linear (e.g., funding costs arising from differential borrowing and lending rates), and second, the case of a linear insurance premium function $f(\pi,\kappa)=-(1-\kappa)q(\pi)$ with premium rate $q(\pi)$ depending on the portfolio composition. In particular, our setting also incorporates nonlinear insurance premium schedules, which help enforce consumer self-selection and reduce insurers' adverse selection risk.\footnote{Usually, insurance pricing considers a uniform price per unit of coverage, typically based on the net expected reservation premium. However, this may lead to various problems, with adverse selection being the most predominant example. Indeed, insurers suffer adverse effects when they decide to extend coverage to individuals whose actual risk is substantially higher than the risk known by the insurer, hence offering insurance at a cost that does not accurately reflect actual risk exposure.}

Most works on dynamic utility maximization with nonlinear portfolio allocation frictions use the method of convex duality, see e.g. \cite{cuocoliu}, \cite{roche2003}, \cite{long2004investment}, \cite{kleinrogers} and \cite{heunis2015}. We do not employ this method since there are not state constraints leading to the HJB variational inequality studied by \cite{mnif2012numerical,mnif2013optimal}. Instead, Lemma \ref{lemma-optcond} below allows us to use directly the HJB partial differential equation associated with the (primal) optimal control problem for constant relative risk aversion (CRRA) power and log-utility preferences. 

The flexibility of the CRRA setting enables us to derive precise theoretical results and allows for a straightforward and accessible sensitivity analysis. Moreover, since the diffusion part of the insurable background risk  is assumed to be stochastically dependent on the asset returns, the obtained optimal decisions can be very different from the decisions that would be recommended if the correlation and/or frictions were ignored, which emphasizes the importance of understanding the direction and degree of dependence. Given the form of the final wealth, our setting can also be seen as a dynamic version of models with stochastic background wealth and insurable multiplicative risk studied by  \cite{tsetlin2005risky}, \cite{franke2006multiplicative},  \cite{franke2011risk}, \cite{franke2018risk}.  Finally, the recent works \cite{bolton2019optimal} and \cite{hong2020mitigating} use DP in similar continuous-time models in which firms buy insurance against idiosyncratic productivity shocks and disaster exposure, respectively, but with no investment allocation frictions. The following are the main contributions of our work.
\begin{enumerate}
	\item We extend the works of \cite{touzi2000optimal} and \cite{lin2012risky} to include nonlinear portfolio allocation frictions. In a significantly general setting, we characterize optimal investment and insurance demand strategies for a representative agent with CRRA preferences and find conditions under which the agent assumes the insurable risk entirely, partially, or purchases total insurance against it. See Theorem \ref{g-theorem} below. 	In particular, we prove that the insurance demand increases with the first-order stochastic dominance of the jump size and the magnitude of the jump arrival rate, see Corollary \ref{FOSD} below.
	
	\item In the case of differential rates for lending and borrowing, this is the first paper -to the best of our knowledge- that studies such frictions jointly with insurance demand for background risk in a dynamic setting. 	We find sufficient conditions (Corollary \ref{cor-diff-rates}) under which the agent responds to the insurable background risk by investing at the (lower) lending risk-free rate $r$, holding only investments in risky assets, or leveraging the risky-asset portfolio by borrowing at the (higher) funding rate $R.$ We find similar results for the case of a large investor with piecewise constant price impact on the drift of the risky asset, see Corollary \ref{large} below.
	
\item We prove a mutual-fund separation result (Theorem \ref{mutual}) for the case of differential rates which shows that optimal portfolio allocations move along one-dimensional segments. In particular, the optimal allocation for a given risk tolerance level can be obtained as a combination of two mutual funds. This brings to light the relevance of our results for CRRA preferences.
\end{enumerate}
The organization of this paper is as follows. In Section 2 we present the dynamics of the risky investments, insurable background risk, nonlinear portfolio allocation frictions $f(\pi,\kappa)$ and the agent's wealth process. We also define the risk-averse utility maximization control problem. In Section 3 we introduce the Hamilton-Jacobi-Bellman equation for the case of CRRA utility. In Section 4 we consider the case  $f(\pi,\kappa)=g(\pi)-p(\kappa)$ with $p$ differentiable, with special focus on the cases $g$ differentiable (for one risky asset) or $g$ piece-wise linear (for multiple risky assets). We present explicit characterization and numerical examples of optimal decisions for agents with exposure to differential rates, and prove a mutual-fund separation Theorem for this case. Finally, in Section 5 we briefly address linear premium functions $f(\pi,\kappa)=-(1-\kappa)q(\pi)$ with premium rate depending on risky asset allocation. In Section 6 we close out the paper with a few conclusions of our work.

\section{Insurable risk model with non-linear portfolio allocation frictions}

Let $T>0$ be a fixed time horizon. At each time $t\in[0,T]$ the representative agent holds a portfolio of investments in $d$ risky assets with price processes $S^1,\ldots,S^d$ and a risk-free money market account $B$ following a Black-Scholes model of the form
\begin{align*}
	dS_t^i&=S_t^i\Bigl[\mu^i\,dt+\sum_{k=1}^d\sigma^{ik}\,dW_t^k\Bigr], \  S_0^i>0, \ \ i=1,\ldots,d\\
	dB_t&=B_tr\,dt, \ B_0=1
\end{align*}
where $W=(W^1,\ldots,W^d)^\top$ is a $d$-dimensional Brownian motion defined on a complete probability space $(\Omega,\Prob,\mathcal{F})$ endowed with a filtration $\mathds F.$ For each $i=1,\ldots,d,$ let $\pi^i_t$ denote the fraction of wealth invested in $i$-th risky asset at time $t\in[0,T]$, so $1-\pi_t^\top\underline 1$ denotes the proportion of wealth invested in the risk-free asset. As usual, we refer to $\pi_t$ as the portfolio proportion process. We assume the return rate of the (self-financed) representative agent's wealth is subject to exogenous shocks modeled by a jump-diffusion process of the form
\[
X_t=b\bar W_t+\sum_{\tau_n\le t} Y_n
\]
where $\bar W$ is a Brownian motion satisfying  $\seq{W^k,\bar W}_t=\rho^{k}t$ with {$\rho^{k}\in[-1,1]$}, $b\in\R,$ and the marked point process $(\tau_n,Y_n)_{n\geq 1}$ is independent of $W$ and $\bar W.$ The correlations $\rho^{k}$ model the dependence between the log-prices of the financial assets and the (Gaussian) fluctuations of the exogenous shocks. In the remainder, we denote $\rho:=(\rho^1,\ldots,\rho^d)^\top\in\R^{d}$ which will be considered as a column vector in the remainder. The representative agent has the possibility of reducing these exogenous shocks on
the wealth process by buying insurance, which in turn causes (endogenous) frictions that depend on the risky asset weights and the insurance strategy. More concretely, let $\kappa=(\kappa_t)_{t\in[0,T]}$ be a $\mathds F$-predictable process with values in $[0,1]$ that represents the fraction of exogenous risk assumed by the agent so that the infinitesimal change of the exogenous shocks process $dX_t$ is reduced to $\kappa_t\,dX_t.$ We assume that there exists a function
\[
f:\R^d\times [0,1]\to \R
\]
which is concave, strictly decreasing in the second variable, so that the wealth process $V^{\pi,\kappa}$ of the agent evolves according to the controlled linear SDE of jump-diffusion type
\begin{equation}\label{eqVnon-linear}
		dV_t=V_{t-}\Bigl\{[r+f(\pi_t,\kappa_t)]\,dt+\pi_t^\top[(\mu-r\underline 1)\,dt+\sigma\,dW_t]-\kappa_t\,dX_t\Bigr\}.
\end{equation}
The function $f(\pi,\kappa)$ represents the (possibly negative) rate of return of an additional endogenous cash flow, relative to the value of investments in the financial market and the cost of insurance. We will focus primarily on two cases:
\begin{itemize}
	\item $f(\pi,\kappa)=g(\pi)-p(\kappa)$ where $g(\pi)$ captures nonlinear investment frictions, usually referred to as margin payment function (see e.g. \cite{cuocoliu}), and $p:[0,1]\to\R_+$ with $p(1)=0$ is the (possibly nonlinear) insurance premium rate function as in \cite{touzi2000optimal}. Nonlinear insurance pricing is a primary mechanism for effecting consumer self-selection and helps reduce adverse selection risk for insurers, see e.g. \cite{schlesinger1983nonlinear}. As for the market frictions, we have for instance the case of higher interest rates for borrowing than lending
	\begin{equation}\label{mpf-dr}
		g(\pi):=-(R-r)(\pi^\top\underline{1}-1)^+, \ \ \pi\in\R^d
	\end{equation}
or, in the one-dimensional case,  price pressure of a large-investor on the risky asset\footnote{Buying the risky asset depresses its expected return, while shorting it has the opposite effect, see e.g. \cite{cuoco1998} or \cite{long2004investment}.}
\[
g(\pi)=\pi[m^+ \mathbf{1}_{\Set{\pi\geq 0}}+m^- \mathbf{1}_{\Set{\pi<0}}], \ \ \pi\in\R
\]
with  $m^+<0<m^-.$

\item $f(\pi,\kappa)=-(1-\kappa)q(\pi)$ with $q$ a positive-valued premium rate that depends on the portfolio weights.
\end{itemize}
We assume $Y_n<1$ for all $n\in\mathds N,$ which guarantees positivity of the wealth process $V^{\pi,\kappa}.$ Indeed, for an initial endowment $x>0,$ the solution to equation (\ref{eqVnon-linear}) is given by
\[
V_t^{\pi,\kappa}=xK_t^{\pi,\kappa}\prod_{\tau_n\le t}(1-\kappa_{\tau_n}Y_{n})
\]
with
\begin{align}
K_t^{\pi,\kappa}&=\exp\biggl(\int_0^t\Bigl[r+f(\pi_s,\kappa_s)+\pi_s^\top(\mu-r\underline 1)\Bigr.\biggr.\nonumber\\
&\biggl.\Bigl.-\frac{1}{2}\Bigl( \bigl\lvert \sigma^\top\pi_s \bigr\rvert^2+(b\kappa_s)^2-2b\kappa_s\pi_s^\top\sigma\rho\Bigr)\Bigr]\,ds+\int_0^t \pi_s^\top\sigma\,dW_s-\int_0^t\kappa_sb\,d\bar W_s\biggr).\nonumber
\end{align}
Hence, our setting can be seen as a dynamic version of models with stochastic background wealth and insurable multiplicative risk, see e.g. \cite{tsetlin2005risky}, \cite{franke2006multiplicative},  \cite{franke2011risk}, \cite{franke2018risk}. Examples of such models, with $K^{\pi}:=K^{\pi,1}$ (no insurance demand) include the following
\begin{enumerate}
\item $K^\pi$ is the profit of a firm that operates internationally, and $1-Y_n$ is the change in the associated exchange rate at time $\tau_n.$

\item $K^{\pi}$ is the nominal wealth, value, or profit of an investment fund, and $1-Y_n$ is the price deflator or purchasing power index reflecting uncertain inflation over the interval $(\tau_{n-1},\tau_n].$ 

\item $K^{\pi}$ is the pretax profits of a firm, and $1-Y_n$ is the firm’s retention rate net of taxes over the period $(\tau_{n-1},\tau_n]$, where tax rates are random due to tax legislation uncertainty.

\item $K^{\pi}$ is the portfolio value of a pension fund, and $1-Y_n$ is the return on a mandatory annuity account, say a default, catastrophe, or mortality-linked security, that rolls over the proceeds from $K^\pi$ after time $\tau_n.$

\item $K^{\pi}$ is the portfolio value of a pension fund, and $1-Y_n$ is the return on a mandatory annuity account, say a default, catastrophe, or mortality-linked security, that rolls over the proceeds from $K^\pi$ after time $\tau_n.$
 \end{enumerate}
Finally, the returns of the risky investment can also be used to model the return rate of capital accumulation, as it is now widely used in macro-finance and economic growth AK models: $\pi$ represents the investment capital ratio, and $g(\pi)$ captures the adjustment and depreciation costs that the firm incurs in the capital investment process, so $K^\pi$ is the total stock of capital of a firm in the AK model, including physical capital as traditionally measured, but also human capital and firm-based intangible capital (such as patents, know-how, brand value, and organizational capital). This is a common assumption in the recent literature on the q theory of investment. The jumps can be used to capture either productivity shocks or disaster risk, see e.g. \cite{bolton2019optimal} and \cite{hong2020mitigating}. 

The following is the formulation of the risk-averse optimization problem. Let $U(x)$ be a utility function satisfying the usual Inada conditions. We denote with $\mathcal A(t,x)$ the set of admissible strategies $(\pi,\kappa)$ for which $\Exp\bigl[U(-V_T^{\pi,\kappa})^+\,\bigl\lvert \,V_t^{\pi,\kappa}=x\bigr]$ is finite. The goal of the representative agent is to maximize the expected final-wealth utility functional $\Exp\left[U(V_T^{\pi,\kappa})\,\bigl\lvert \, V_t^{\pi,\kappa}=x\bigr.\right]$ over all admissible strategies $(\pi,\kappa)\in\mathcal A(t,x).$ For this, we define the time-dependent optimal value function
\begin{equation}\label{opt-vf}
\vartheta(t,x):=\sup_{(\pi,\kappa)\in\calA(t,x)}\Exp\left[U(V_T^{\pi,\kappa})\,\bigl\lvert\,V_t^{\pi,\kappa}=x\bigr.\right].
\end{equation}

\section{HJB equation for CRRA preferences}

In what follows, we assume the random variables $\Set{Y_n}_{n\in\mathds N}$ are i.i.d. and the counting process $N_t=\sum_{n=1}^\infty\mathbf 1_{\Set{\tau_n\le t}}$ is a Poisson process with intensity $\lambda>0.$ If the optimal value function (\ref{opt-vf}) is sufficiently differentiable, it satisfies the non-linear second-order integro-differential equation, usually referred to as Hamilton-Jacobi-Bellman (HJB) equation
\[
\frac{\partial \vartheta}{\partial t}+\sup_{{\substack{\pi\in\R^d\\\kappa\in[0,1]}}}[\mathcal L^{\pi,\kappa}\vartheta](t,x)=0
\]
with final condition
\[
\vartheta(T,x)=U(x)
\]
where, for each $\pi\in\R^d$ and $\kappa\in[0,1],$ $\mathcal L^{\pi,\kappa}$ is the second-order integro-differential operator
\begin{align*}
	[\mathcal L^{\pi,\kappa}\vartheta]&(x)= \,x\Bigl[r+f(\pi,\kappa)+\pi^\top\bigl(\mu-r\underline 1\bigr)\Bigr]\frac{\partial \vartheta}{\partial x}\\
	&+\frac{x^2}{2}\left[\bigl\lvert \sigma^\top\pi \bigr\rvert^2+(b\kappa)^2-2b\kappa b\pi^\top\sigma\rho\right]\frac{\partial^2\vartheta}{\partial x^2}
	+\lambda\Exp\left[\vartheta(x(1-\kappa Y ))\right]-\lambda\vartheta(x).
\end{align*}
Conversely, the so-called verification Theorem links the solution of the above HJB equation with sufficient conditions for existence of optimal strategies. Suppose further that agent has a von Neumann–Morgenstern CRRA utility function of the form
\[
U_\eta(x)=
\left\{
\begin{array}{ll}
	\frac{x^{1-\eta}}{1-\eta}, &  \eta\in(0,+\infty)\setminus\Set{1}, \\\\
	\ln x, & \eta =1.
\end{array}
\right.
\]
We employ the following guess for the value function
\[
v(t,x)=
\left\{
\begin{array}{ll}
	\theta_\eta(t)\frac{x^{1-\eta}}{1-\eta}, & \eta\in(0,+\infty)\setminus\Set{1}, \\\\
	\theta_1(t)+\ln x, & \eta =1,
\end{array}
\right.
\]
with $\theta_\eta:[0,T]\to[0,\infty)$ differentiable and positive-valued. Substituting in the HJB equation, we see that the maximization problem in the HJB equation is reduced to maximizing $f(\pi,\kappa)+H(\pi,\kappa;\eta)$ over $(\pi,\kappa)\in\R^d\times[0,1]$ with
\[
H(\pi,\kappa;\eta):=\pi^\top (\mu-r\underline 1)
-\frac{\eta}{2}\left[\bigl\lvert \sigma^\top \pi \bigr\rvert^2+(b\kappa)^2-2b\kappa \pi^\top\sigma\rho\right]+\lambda\Exp[U_\eta(1-\kappa Y)].
\]
For the case $\eta\neq 1$ the function $\theta_\eta$ is the solution to the final-value ODE
\[
\frac{1}{1-\eta}\theta_\eta'(t)+\theta_\eta(t)[r-\lambda+f(\hpi,\hat\kappa)+H(\hpi,\hat\kappa;\eta)]=0, \ \ \theta_1(T)=1
\]
where
\[
(\hpi,\hat\kappa)=\arg\max_{{\substack{\pi\in\R^d\\\kappa\in[0,1]}}}f(\pi,\kappa)+H(\pi,\kappa;\eta).
\]
In the log-utility case $\eta=1$ we have $\theta_1(t)=-[r-f(\hpi,\hat\kappa)+H(\hpi,\hat\kappa;1)]t.$ Since $f$ is not necessarily differentiable, we will employ the following  `convex conjugate'
\[
\tilde f(\zeta,\gamma):=\sup_{\substack{\pi\in\R^d\\\kappa\in[0,1]}}[f(\pi,\kappa)+\pi^\top\zeta+\kappa\gamma], \ \ \zeta\in\R^d, \ \gamma\in\R
\]
and the effective domain $\N:=\Set{\tilde{f}<+\infty}.$ We have the following result
\begin{lemma}\label{lemma-optcond}
Let $(\hpi,\hat\kappa)\in\R^d\times[0,1]$ be such that
\begin{equation}\label{optcond-zeta-pikappa}
f(\hpi,\hat\kappa)+\hat\pi^\top \nabla_\pi H({\hat\pi,\hat\kappa})+\hat\kappa\partial_\kappa H({\hat\pi,\hat\kappa})=\tilde f(\nabla_\pi H({\hat\pi,\hat\kappa}),\partial_\kappa H({\hat\pi,\hat\kappa}))
\end{equation}
where
\begin{equation}\label{grad_partialH}
\begin{split}
	\nabla_\pi H&=\mu-r\underline 1-\eta\sigma[\sigma^\top\pi-\rho b\kappa],\\	
	{\partial_\kappa H}&=\eta b[\pi^\top\sigma\rho -b\kappa]-\lambda\Exp\Bigl[\frac{Y}{(1-\kappa Y)^\eta}\Bigr].
\end{split}
\end{equation}
are the gradient and partial derivative of $H$ with respect to $\pi\in\R^d$ and $\kappa\in[0,1]$ respectively. Then the pair $(\hpi,\hat\kappa)$ maximizes $f+H.$
\end{lemma}
\begin{proof}
Let $(\pi,\kappa)\in\R^d\times[0,1]$ be fixed. Then,
\begin{align*}
	&f(\hpi,\hat\kappa)+H(\hpi,\hat\kappa;\eta)\\
	&=\tilde f(\nabla_\pi H({\hat\pi,\hat\kappa}),\partial_\kappa H({\hat\pi,\hat\kappa}))+H(\hpi,\hat\kappa;\eta)-\hat\pi^\top \nabla_\pi H({\hat\pi,\hat\kappa})-\hat\kappa\partial_\kappa H({\hat\pi,\hat\kappa})\\
	&\geq f(\pi,\kappa)+\pi^\top \nabla_\pi H({\hat\pi,\hat\kappa})+\kappa \partial_\kappa H({\hat\pi,\hat\kappa})+H(\hpi,\hat\kappa;\eta)-\hat\pi^\top \nabla_\pi H({\hat\pi,\hat\kappa})-\hat\kappa \partial_\kappa H({\hat\pi,\hat\kappa})\\
	&=f(\pi,\kappa)+H(\hpi,\hat\kappa;\eta)+(\pi-\hat\pi)^\top \nabla_\pi H({\hat\pi,\hat\kappa})+(\kappa-\hat\kappa){\partial_\kappa H}({\hat\pi,\hat\kappa})\\
	&\geq f(\pi,\kappa)+H(\pi,\kappa;\eta).
\end{align*}
The last inequality follows from the concavity of $H.$	
\end{proof}

\section{$f(\pi,\kappa)=g(\pi)-p(\kappa)$ with $p$ differentiable}

Let us consider first the case $f(\pi,\kappa)=g(\pi)-p(\kappa)$ with $g$ concave and $g(\underline 0)=0$, and $p$ a differentiable convex premium function on $[0,1]$ with $p(1)=0.$ 
For $\gamma\in\R,$ the first-order Karush–Kuhn–Tucker (KKT) optimality conditions for the maximization problem $\sup_{\kappa\in[0,1]}-p(\kappa)+\gamma \kappa$ are
\begin{align*}
-\gamma+p'(\kappa)+\phi_1-\phi_2&=0\\
\phi_1(\kappa-1)&=0\\
\phi_2\kappa&=0
\end{align*}
with Lagrange multipliers $\phi_1,\phi_2\geq 0.$ Then, we deduce that 
\[
\tilde f(\zeta,\gamma)=\tilde g(\zeta)+
\begin{cases}
-p((p')^{-1}(\gamma))+(p')^{-1}(\gamma), \ \ &\mbox{ if } \ \gamma\in\range p',\\
-p(0), \ \ &\mbox{ if } \ \gamma< p'(0),\\
\gamma, \ \ &\mbox{ if } \ \gamma> p'(1).
\end{cases}
\]
where
\[
\tilde g(\zeta):=\sup_{\pi\in\R^d}g(\pi)+\pi^\top\zeta, \ \zeta\in\R^d.
\]
A straightforward application of Lemma \ref{lemma-optcond} yields the following result, which provides a characterization of sufficient conditions for existence of optimal investment and insurance demand choices. Below, in the remainder of this section, we present some examples in which the result can be used to obtain explicit solutions.
\begin{theorem}[Characterization of optimal portfolio and insurance demand]\label{g-theorem}
\phantom{AA}

\begin{enumerate}
		\item Suppose there exists $(\hpi,\hat\kappa)$ satisfying the system of $d+1$  equations
		\begin{align}
			\tilde g(\nabla_\pi H(\hat\pi,\hat\kappa))&=g(\hat\pi)+\hat\pi^\top\nabla_\pi H(\hat\pi,\hat\kappa)\label{tilde-g-zeta}\\
			\partial_\kappa H(\hpi,\hat\kappa)&=p'(\hat\kappa)\label{gtilde-pprime-2}
		\end{align}
then $(\hpi,\hat\kappa)$ is optimal.

\item If there exists $\hpi$ satisfying
\begin{equation}\label{partialH-pizero}
\partial_\kappa H(\hpi,0)=\eta b\rho^\top\sigma^\top\hpi-\lambda \Exp[Y]\leq p'(0)\end{equation}
and (\ref{tilde-g-zeta}) with $\kappa=0,$ that is,
\[
\tilde g(\mu-r\underline 1-\eta\sigma\sigma^\top\hpi)=g(\hpi)+\hpi^\top(\mu-r\underline 1-\eta\sigma\sigma^\top\hpi)
\]
then $(\hpi,0)$ is optimal.

\item If there exists $\hpi$ satisfying
\begin{equation}\label{partialH-pione}
	\partial_\kappa H(\hpi,1)=\eta b[\rho^\top\sigma^{\top}\hpi -b] -\lambda\Exp\Bigl[\frac{Y}{(1-Y)^\eta}\Bigr]>p'(1)
\end{equation}
and (\ref{tilde-g-zeta}) with $\kappa=1,$ that is,
\[
\tilde g(\mu-r\underline 1-\eta\sigma[\sigma^\top\hpi-\rho b])=g(\hpi)+\pi^\top(\mu-r\underline 1-\eta\sigma[\sigma^\top\hpi-\rho b])
\]
then $(\hpi,1)$ is optimal.
\end{enumerate}
\end{theorem}

\begin{example}
Let $p(\kappa)=q(1-\kappa)^\delta$ with $q>0$ and $\delta\geq 1.$ Then condition
(\ref{partialH-pizero}) in case 2 reads $\lambda \Exp[Y]-\eta b\rho^\top\sigma^\top\hpi\geq \delta q$. In particular, for sufficiently high values of the (actuarially fair) reservation premium $\lambda\Exp[Y],$ the agent insures completely against the exogenous shocks in the return rate of the wealth process. In the linear case $(\delta=1)$ condition
(\ref{partialH-pione}) in case 3 reads
\[
\lambda\Exp\Bigl[\frac{Y}{(1-Y)^\eta}\Bigr]-\eta b[\rho^\top\sigma^{\top}\hpi -b] <q
\]
so the agent assumes completely the risk of negative jumps if the return rate for premium rate $q$ is sufficiently high.
\end{example}

\subsection{$g$ differentiable}
Suppose $d=1$, $g\in\mathcal C^2(\R)$ and $g''<0$ so that $g$ is strictly concave. Using first and second order optimality conditions, we obtain
\[
\tilde g(\zeta)=g((g')^{-1}(-\zeta))+\zeta (g')^{-1}(-\zeta), \ \ \zeta\in\mathrm{Range } (-g').
\]
and $\mathcal N=\mathrm{Range } (-g')\times\R,$ so the sufficient condition (\ref{tilde-g-zeta}) for optimality of $(\pi,\kappa)$ reads $\pi=(g')^{-1}(-\mu+r+\eta\sigma[\sigma\pi-\rho b\kappa]),$ that is, 
\[
\mu-r+\eta\sigma\rho b \kappa=\eta\sigma^2\pi-g'(\pi).
\]
By writing the right side as $Q(\pi):=\eta\sigma^2\pi -g'(\pi),$ optimality condition (\ref{tilde-g-zeta}) can be rewritten as 
\begin{equation}\label{eq-pikappa-fprime}
	Q(\hpi)=\mu-r+\eta\sigma\rho b \hat\kappa.
\end{equation}
Notice $Q$ is strictly increasing since $Q'=\eta\sigma^2-g''>0.$ Then we have the following result
\begin{corollary}
\begin{enumerate}
    \item If there exists $\hat\kappa\in[0,1]$ such that $\mu-r+\eta\sigma\rho b \hat\kappa\in {\rm Range}(Q)$ and 
	\begin{equation}\label{optcond-Qpprime}
	\eta b[\sigma\rho Q^{-1}(\mu-r+\eta\sigma\rho b\hat\kappa)-b \hat\kappa]-\lambda\Exp\Bigl[\frac{Y}{(1-\hat\kappa Y)^\eta}\Bigr]-p'(\hat\kappa)=0		
	\end{equation}
	then $(\hpi,\hat\kappa)$ is optimal with $\hpi:=Q^{-1}(\mu-r+\eta\sigma\rho b\hat\kappa).$
	
\item Suppose that $\rho>0$ holds if and only if $\mu<r+Q\bigl(\frac{1}{\eta \sigma b\rho}[\lambda\Exp(Y)+p'(0)]\bigr),$ and $\mu-r\in {\rm Range}(Q).$ Then the pair $(\hpi,\hat\kappa)=(Q^{-1}(\mu-r),0)$ is optimal.

\item Suppose that $\rho>0$ holds if and only if
\[
\mu-r+\eta\rho b\sigma>Q\Bigl(\frac{1}{\eta \rho b\sigma}\Bigl\{\lambda\Exp\Bigl[\frac{Y}{(1-Y)^\eta}\Bigr]+\eta b^2+p'(1)\Bigr\}\Bigr)
\]
and $\mu-r+\eta\rho b\sigma\in{\rm Range}(Q),$ then  $(\hpi,\hat\kappa)=(Q^{-1}(\mu-r+\eta\rho b\sigma),1)$ is optimal.
\end{enumerate}
\end{corollary}
\begin{proof}
Inserting $\pi=Q^{-1}(\mu-r+\eta\sigma\rho b\kappa)$ into (\ref{gtilde-pprime-2})  yields (\ref{optcond-Qpprime}), and part \emph{1} follows. Parts \emph{2} and \emph{3} follow similarly from inserting $\pi=Q^{-1}(\mu-r+\eta\sigma\rho b\kappa)$ with $\kappa=0$ and $\kappa=1$ into (\ref{partialH-pizero}) and (\ref{partialH-pione}) respectively.
\end{proof}

\begin{remark}\label{exist-kappah}
Note that $g''\circ Q^{-1}<0<\eta\sigma^2(1-\rho^2).$ Then $\eta(\sigma\rho)^2<\eta\sigma^2-g''\circ Q^{-1}$ and
\[
1-\eta(\sigma\rho)^2(Q^{-1})'=1-\frac{\eta(\sigma\rho)^2}{\eta\sigma^2-g''\circ Q^{-1}}>0.
\]
It follows that
\begin{equation}\label{h-kappa}
h(\kappa):=\eta b[\sigma\rho Q^{-1}(\mu-r+\eta\sigma\rho b\kappa)-b \kappa]-\lambda\Exp\Bigl[\frac{Y}{(1-\kappa Y)^\eta}\Bigr]-p'(\kappa)	
\end{equation}
is strictly decreasing since
\[
h'(\kappa)=-\eta\left\{b^2[1-\eta(\sigma\rho)^2(Q^{-1})'(\mu-r+\eta\sigma\rho b\kappa)]+\lambda\Exp\Bigl[\frac{Y^2}{(1-{{\kappa}}Y)^{1+\eta}}\Bigr]\right\}-p''(\kappa)
\]
is strictly negative. Hence, there exists an unique optimal $\hat\kappa$ for case 1 whenever $h(0)>0>h(1)$ i.e. zero is between the values
\begin{align*}
&\eta b[\sigma\rho Q^{-1}(\mu-r+\eta\sigma\rho b)-b ]-\lambda\Exp\Bigl[\frac{Y}{(1- Y)^\eta}\Bigr]-p'(1)\\
&\mbox{ and } \ \ \eta b\sigma\rho Q^{-1}(\mu-r)-\lambda\Exp[Y]-p'(0).    
\end{align*}
\end{remark}
Deterioration in background wealth may encompass more complicated changes in the model parameters or distribution of the negative jumps. The previous observation can be used to study comparative static properties for optimal pairs $(\hpi,\hat\kappa)$ with respect to the model parameters and first-order stochastic dominance deterioration of insurable risk. Recall that a random variable $Y$ has first-order stochastic dominance over random variable $\tilde Y$ if the cumulative distribution functions satisfy $F_{Y}\le F_{\tilde Y}.$
\begin{corollary}\label{FOSD}
$\hat\kappa$ decreases with the marginal cost $p',$ jump arrival rate $\lambda$  and   first-order stochastic dominance of  the jump size $Y.$  Both $\hpi$ and $\hat\kappa$ increase (resp. decrease) with the premium risk $\mu-r$ if $\rho>0$ (resp. $<0$). 
\end{corollary}
\begin{proof}
If $Y$ dominates $\tilde Y$ in the sense of first-order stochastic dominance, then
\[
\Exp\Bigl[\frac{Y}{(1-\kappa Y)^\eta}\Bigr]\geq \Exp\Big[\frac{\tilde Y}{(1-\kappa \tilde Y)^\eta}\Big]
\]
since the function $\psi(y)=y(1-\kappa y)^{-\eta}$ is increasing in $y\in (0,1),$
see e.g.  \cite[Ch. 2]{eeckhoudt2011economic}. The assertions follow easily from (\ref{optcond-Qpprime}) and the increasing monotonic behavior of $Q^{-1}.$
\end{proof}
An increase in the first-order stochastic dominance of  the jump size $Y,$ the magnitude of the jump arrival rate $\lambda$ or the marginal cost $p',$ is sufficient to guarantee a decrease in the fraction of insurable risk assumed by the agent. The relation $\hpi=Q^{-1}(\mu-r+\eta\sigma\rho b\hat\kappa)$ shows that such an increase also results in a decrease (resp. increase) of the risky asset allocation if the correlation between the exogenous adverse shocks and financial log-returns is negative (resp. positive). In particular, it is beneficial to invest more in the risky asset if the correlation is negative, as it serves as a hedging strategy against adverse shocks. 

In particular, in the case of no-frictions ($g=0$) it is easy to see that $\hat\kappa$ decreases with $\eta,$ and so does $\hpi$ whenever $\rho>0,$ since
\[
\frac{\partial h}{\partial \eta}=\frac{-\frac{\partial h}{\partial \eta}}{\frac{\partial h}{\partial \kappa}}<0.
\]
A similar argument yields that $\hat\kappa$ increases with correlation $\rho$ if it is larger than $-\frac{\mu-r}{2\eta b\sigma}.$ Note this inequality holds for all $\rho\in[-1,1]$ if $\mu<r+2b\eta \sigma.$ That is, if the mean rate of return 
is sufficiently low, then the insured fraction of jump risk decreases with the correlation.
\subsection{$g$ piece-wise linear}
\subsubsection{Differential rates for borrowing and lending}\label{exdiffrates-2}
We now consider the margin payment function for the case of an interest rate that is higher for borrowing than for investing
\[
g(\pi):=-(R-r)(\pi^\top\underline{1}-1)^+
\]
This (piece-wise linear) margin payment function is not differentiable, yet we may characterize $\tilde g$ rather easily. Indeed, the map
\[
\R\ni\pi\mapsto g(\pi)+\pi^\top\zeta=
\left\{
\begin{array}{ll}
	\pi^\top\zeta, & \ \pi^\top\underline 1\le 1,\\\\
	\pi^\top[\zeta-(R-r)\underline 1]+(R-r), &  \ \pi^\top\underline 1> 1,
\end{array}
\right.
\]
attains a finite maximum value if and only if $0\le\zeta^1=\zeta^2=\cdots=\zeta^d\le R-r,$ see Section 6.8 of \cite{karatzas1998methods}. Then, the effective domain of $\tilde f$ is
\[
\mathcal N=\Set{(\xi-r)\underline 1\in\R^d:\xi\in[r,R]}\times\R
\]
and we have $\tilde g((\xi-r)\underline 1)=\xi-r.$ Solutions to equation (\ref{tilde-g-zeta}) can be singled out as follows
\begin{itemize}
	\item $\pi^\top \underline 1<1$ and $\xi=r,$ i.e. $\eta\sigma(\sigma^\top\pi-\rho b\kappa)=\mu-r\underline 1$
	
	\item $\pi^\top \underline 1>1$ and $\xi=R,$ i.e. $\eta\sigma(\sigma^\top\pi-\rho b\kappa)=\mu-R\underline 1$
	
	\item $\pi^\top \underline 1=1$ and $\eta\sigma(\sigma^\top\pi-\rho b\kappa)=\mu-\xi\underline 1$ for some $\xi\in[r,R].$
\end{itemize}
Using this in conjunction with Thoerem \ref{g-theorem} we get the following result, which generalizes Section 5.2 of \cite{cuocoliu}\footnote{Other similar piece-wise linear margin payment functions related to funding costs (e.g., margin requirements for leveraged positions, offsetting of positions in risky assets, short-sale constraints with negative rebate rates, etc.) can be considered, see e.g. \cite{cuocoliu}, \cite{bielecki2015valuation}.} to the case in which the agent can insure against the exogenous shocks. In what follows, we assume $\sigma$ is invertible, so that $(\sigma\sigma^\top)^{-1}=(\sigma^\top)^{-1}\sigma^{-1}$ and $\sigma^\top(\sigma\sigma^\top)^{-1}=\sigma^{-1}.$
\begin{corollary}\label{cor-diff-rates}
	Suppose for each $\xi\in[r,R]$ there exists $\kappa(\xi)\in[0,1]$ solution to $h(\kappa;\xi)=\underline 0$ with
\[
h(\kappa;\xi):=b\rho^\top\sigma^{-1}({\mu}-\xi\underline 1)-\eta b^2(1-\abs{\rho}^2)\kappa-\lambda\Exp\Bigl[\frac{Y}{(1-\kappa Y)^\eta}\Bigr]-p'(\kappa)
\]	
	and set
	\[
	\pi(\xi):=\frac{1}{\eta}(\sigma\sigma^\top)^{-1}({\mu}-\xi\underline 1)+(\sigma^\top)^{-1}\rho b\kappa(\xi).
	\]
	We have the following
	\begin{enumerate}
		\item[i.] If $\pi(r)^\top\underline 1<1,$ then $\hat\pi=\pi(r)$ and $\hat\kappa=\kappa(r)$ are optimal. In this case, the agent invests only own funds in the financial assets. 
		
		\item[ii.] If $\pi(R)^\top\underline 1> 1,$ then $\hat\pi:=\pi(R)$ and $\hat\kappa=\kappa(R)$ are optimal. In this case, the agent leverages the position in risky assets by borrowing from the risk-free money market account. 
		
		\item[iii.] If there exists $\xi^*\in[r,R]$ such that
		\begin{equation}\label{pi-xi}
			\pi(\xi^*)^\top \underline 1=1,
		\end{equation}
		then $\hat\pi=\pi(\xi^*)$ and $\hat\kappa=\kappa(\xi^*)$ are optimal. In this case, the agent holds a portfolio consisting of only risky assets.
		
		\item[iv.] Suppose $\lambda \Exp[Y]+p'(0)\geq b\rho^\top\sigma^{-1}(\mu-r\underline 1)$ and $\underline 1^\top(\sigma\sigma^\top)^{-1}(\mu-r\underline 1)\le \eta.$ Then the pair $\hpi=\frac{1}{\eta}(\sigma\sigma^\top)^{-1}(\mu-r\underline 1),$ $\hat\kappa=0$ is optimal. In particular, the agent insures totally against the adverse shocks and invests a positive amount in the risk-free money market account.  
		
		\item[v.] Suppose $\lambda \Exp[Y]+p'(0)\geq b\rho^\top\sigma^{-1}(\mu-R\underline 1)$ and $\underline 1^\top(\sigma\sigma^\top)^{-1}(\mu-R\underline 1)> \eta.$ Then the pair $\hpi=\frac{1}{\eta}(\sigma\sigma^\top)^{-1}(\mu-R\underline 1),$ $\hat\kappa=0$ is optimal. In particular, the agent insures totally against the adverse shocks and leverages the portfolio of risky assets by borrowing from the risk-free money market account.  
		
		\item[vi.] Suppose $\lambda\Exp\Bigl[\frac{Y}{(1-Y)^\eta}\Bigr]-b\rho^\top\sigma^{-1}(\mu-r\underline 1)+\eta b^2(1-\abs{\rho}^2)+p'(1)<0$ and
		\[
		\underline 1^\top\Bigl[\frac{1}{\eta}(\sigma\sigma^\top)^{-1}(\mu-r\underline 1)+b(\sigma^\top)^{-1}\rho\Bigr]\le 1.
		\]
		Then the pair $\hpi=\frac{1}{\eta}(\sigma\sigma^\top)^{-1}(\mu-r\underline 1)+b(\sigma^\top)^{-1}\rho,$ $\hat\kappa=1$ is optimal. In this case the agent assumes totally the insurable background risk and invests a positive amount in the risk-free money market account.  
		
		\item[vii.] Suppose $\lambda\Exp\Bigl[\frac{Y}{(1-Y)^\eta}\Bigr]-b\rho^\top\sigma^{-1}(\mu-R\underline 1)+\eta b^2(1-\abs{\rho}^2)+p'(1)<0$ and
		\[
		\underline 1^\top\Bigl[\frac{1}{\eta}(\sigma\sigma^\top)^{-1}(\mu-R\underline 1)+b(\sigma^\top)^{-1}\rho\Bigr]> 1.
		\]
		Then the pair $\hpi=\frac{1}{\eta}(\sigma\sigma^\top)^{-1}(\mu-R\underline 1)+b(\sigma^\top)^{-1}\rho,$ $\hat\kappa=1$ is optimal. In this case the agent assumes totally the insurable background risk and leverages the investment in risky assets by borrowing from the risk-free money market account.  
		
	\end{enumerate}
\end{corollary}
\begin{proof}
$\tilde g(\nabla_\pi H(\pi,\kappa))$ is finite if  $\nabla_\pi H(\pi,\kappa) =(\xi-r)\underline 1$ for some $\xi\in[r,R],$ in which case $(\ref{tilde-g-zeta})$ reads
\[
\xi-r=-(R-r)(\pi^\top\underline 1-1)^++\pi^\top\nabla_\pi H(\pi,\kappa).
\]
We can single out three cases: either $\pi^\top\underline 1<1, \pi^\top\underline 1>1$ or $\pi^\top\underline 1=1.$ For each case it suffices to have, respectively, $\nabla_\pi H(\pi,\kappa)=\underline 0,$ $\nabla_\pi H(\pi,\kappa)=(R-r)\underline 1$ or $\nabla_\pi H(\pi,\kappa)=(\xi^*-r)\underline 1$ for some $\xi^*\in[r,R].$ Solving these (linear) equations for $\pi$ as a function of $\kappa,$ and inserting into equation (\ref{gtilde-pprime-2}) yields $h(\kappa;\xi)=\underline 0$ with $\xi=r,R$ and $\xi^*$ respectively, and cases \emph{i, ii} and \emph{iii} follow. Cases \emph{iv - vii} follow similarly from \emph{2} and \emph{3} in Theorem \ref{g-theorem}.
\end{proof}

Note that $\pi(\xi)$ is the optimal Merton proportion plus a hedging component against the insurable exogenous jump risk. That is, the agent uses the optimal demand in risky assets to manage its exposure to both financial and background risk. The correlations also contribute to the diversification effect. Moreover, note that if (\ref{pi-xi}) holds, no wealth is invested in the riskless market account. In this case, $\xi$ can be interpreted as a shadow risk-free interest  rate implied by the optimal policy which lies between the riskless lending and borrowing rates and conforms to the risk preferences of the firm.



\begin{remark}
As in Remark \ref{exist-kappah} $\kappa(\xi)$ exists and is unique if $h(0;\xi)>0>h(1;\xi).$ This is equivalent to the following
\begin{equation}\label{cond-exist-kappa-xi}
\lambda\Exp[Y]+p'(0)<b\rho^\top\sigma^{-1}(\mu-\xi\underline 1)<\eta b^2(1-\abs{\rho}^2)+\lambda\Exp\Bigl[\frac{Y}{(1-Y)^\eta}\Bigr]+p'(1).
\end{equation}
Moreover, implicit differentiation of $h(\kappa(\xi);\xi)=0$ gives $\frac{\partial h}{\partial \kappa}\kappa'(\xi)+\frac{\partial h}{\partial \xi}=0,$ hence
	\[
	\kappa'(\xi)=\frac{-\frac{\partial h}{\partial \xi}}{\frac{\partial h}{\partial \kappa}}=\frac{-b\rho^\top\sigma^{-1}\underline 1}{\eta\left[b^2(1-\abs{\rho}^2)+\lambda\Exp\Bigl(\frac{Y^2}{[1-{{\kappa}}Y]^{1+\eta}}\Bigr)\right]+p''(\kappa)}
	\]
	and
	\begin{align*}
		\frac{d}{d\xi}[\pi(\xi)^\top\underline 1]&=-\frac{\abs{(\sigma^\top)^{-1}\underline 1}^2}{\eta}+bk'(\xi)\rho^\top\sigma^{-1}\underline 1\\
		&=\frac{-1}{\eta}\left\{\abs{(\sigma^\top)^{-1}\underline 1}^2+\frac{(b\rho^\top\sigma^{-1}\underline 1)^2}{\left[b^2(1-\abs{\rho}^2)+\lambda\Exp\Bigl(\frac{Y^2}{[1-{{\kappa}}Y]^{1+\eta}}\Bigr)\right]+\eta p''(\kappa)}\right\}<0
	\end{align*}
	that is $\pi(\xi)^\top\underline 1$ is strictly decreasing on its domain. Then, if  (\ref{cond-exist-kappa-xi}) holds for some $\xi^*$ with $\pi(\xi^*)^\top\underline 1<1,$  we can ensure existence of $r^*$ such that $\pi(r^*)^\top\underline 1=1,$ and the optimality criteria of Corollary \ref{cor-diff-rates} simplifies into
	\[
	(\hpi,\hat\kappa)
	=\left\{
	\begin{array}{ll}
		(\pi(r),\kappa(r)), & \ r>r*\\\\
		(\pi(R),\kappa(R)), & \ R<r*\\\\
		(\pi(r^*),\kappa(r^*)), & \ r\le r^*\le R
	\end{array}
	\right.
	\]
\end{remark}
\begin{example}[Beta distribution with linear premium function]
In order to illustrate the previous results, we consider the case of linear premium rate function $p(\kappa)=(1-\kappa)q$ with premium rate $q>0$ and $Y_n\sim$ Beta$(\alpha,\beta)$ with density function
\[
f_Y(y)=\frac{\Gamma(\alpha+\beta)}{\Gamma(\alpha)\Gamma(\beta)}y^{\alpha-1} (1-y)^{\beta-1}, \ \ y\in [0,1].
\]
This distribution is very flexible, as the density $f_Y(y)$ can be U-shaped with asymptotic ends, bell-shaped, strictly increasing/decreasing or even straight lines, depending on the parameters $\alpha,\beta.$ For this distribution, we have the following Euler-type identity
\begin{equation}\label{kappa-eta-beta-hyper}
\Exp\Bigl[\frac{Y}{(1-\kappa Y)^\eta}\Bigr]={}_{2}{F}_{1}(\eta,\alpha+1;\alpha+\beta+1,\kappa)\frac{\alpha}{\alpha+\beta}, \ \ \kappa\in [0,1)
\end{equation}
where ${}_{2}{F}_{1}$ is the Gaussian (or ordinary) hypergeometric function. The identity does not hold for $\kappa=1.$ However, we may calculate the integral for this value directly using the definition of the Beta function in terms of Gamma functions
\[
\Exp\Bigl[\frac{Y}{(1- Y)^\eta}\Bigr]=\frac{\alpha\Gamma(\alpha+\beta)\Gamma(\beta-\eta)}{\Gamma(\beta)\Gamma(\alpha+\beta+1-\eta)}, \ \ \eta <\beta.
\]
For $\kappa=0$ we have $\Exp[Y]=\frac{\alpha}{\alpha+\beta},$ so inequalities (\ref{cond-exist-kappa-xi}) read
\[
\frac{\lambda\alpha}{\alpha+\beta}<q+b\rho^\top\sigma^{-1}(\mu-\xi\underline 1)<\eta b^2(1-\abs{\rho}^2)+\frac{\lambda\alpha\Gamma(\alpha+\beta)\Gamma(\beta-\eta)}{\Gamma(\beta)\Gamma(\alpha+\beta+1-\eta)}.
\]
We present some numerical examples with two risky assets ($d=2$) and two sets of parameters $A_1$ and $A_2.$ For both sets we assume $\alpha=2,$ $\beta=8$, $\lambda=0.25,$
\[
\rho=\begin{pmatrix}0.2\\-0.3\end{pmatrix} \ \ \mbox{ and }  \ \ \sigma=\begin{pmatrix} 0.25 & 0 \\ 0.32s & 0.32\sqrt{1-s^2} \end{pmatrix}.
\]
The remaining parameters are as follows
\begin{table}
	\centering
	\begin{tabular}{|c|c|c|}
		\hline
		Parameter & $A_1$ & $A_2$\\
\hline
		$q$ & $0.3$ & $0.8$ \\
		$r$ & $2\%$ & $3\%$ \\
		$R$ & $6\%$ & $10\%$ \\
		$b$ & $0.4$ & $0.6$ \\
		$s$ & $0.25$ & $0.05$ \\
		$\mu$ & $\displaystyle \begin{pmatrix} 8\% \\ 10\% \end{pmatrix}$ & $\displaystyle \begin{pmatrix} 16\% \\ 8\% \end{pmatrix}$ \\
		\hline
	\end{tabular}
	\caption{Parameter sets $A_1$ and $A_2.$}
\end{table}
Figures \ref{fig:pi_12_eta_A} and \ref{fig:pieta_A} show the plots of $\hat\pi^i,$ for $i=1,2$ and $\hat\pi^\top\underline 1=\hat\pi^1+\hat\pi^2$ (optimal total position in risky assets) as functions of risk aversion coefficient $\eta,$ respectively.  In Figure \ref{fig:pi_12_eta_A} we see that for the parameter set $A_1$ both $\hpi^1$ and $\hpi^2$ decrease for small values of $\eta$, whereas for set $A_2$ the optimal weight $\hpi^2$ is mostly negative and increasing. This is due mostly to asset $1$ in set $A_2$ having an expected return ($\mu_1=16\%$) greater than the borrowing rate $R=10\%$ and the expected return of asset 2 $(\mu_2=8\%$), so the agent leverages its long position in asset 1 by borrowing from the risk-free market account and short-selling asset 2. 
\begin{figure}[H]
	\centering
	\includegraphics[scale=0.4]{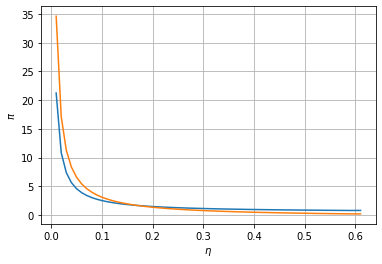}
	\includegraphics[scale=0.4]{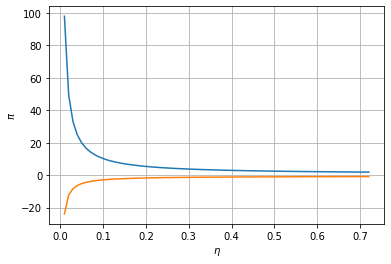}
	\caption{$\hpi^1(\eta)$ in blue and $\hpi^2(\eta)$ in orange, for parameter sets $A_{1}$ (left) and $A_2$ (right).}\label{fig:pi_12_eta_A}
\end{figure}
Figure \ref{fig:pieta_A} illustrates the first three cases of Corollary \ref{cor-diff-rates}. As expected,  $\hat\pi^\top\underline 1=\hat\pi^1+\hat\pi^2$ decreases as $\eta<\beta$ increases. Moreover, there exist $0<\eta_R<\eta_r$ such that $\hpi^\top\underline 1>1$ (leveraged portfolio, case ii of Corollary \ref{cor-diff-rates}) decreases strictly on $(0,\eta_R],$ equals 1 (risky assets only, case iii of Corollary \ref{cor-diff-rates}) over $[\eta_R,\eta_r]$ and then again decreases strictly and tends to $0$ on $[\eta_r,\infty).$ In particular, for $\eta\in [\eta_R,\eta_r]$ the agent's portfolio is fully invested in risky assets. For set $A_1$ we have  $\eta_r =  1.47$ and $\eta_R =  0.6$, for set $A_2$ we get $\eta_r =  2.35$ and $\eta_R =  0.71$. 
\begin{figure}[H]
	\centering
	\includegraphics[scale=0.4]{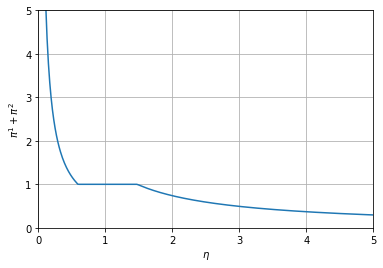}
	\includegraphics[scale=0.4]{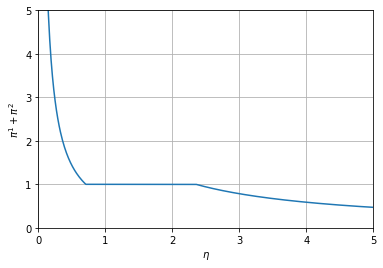}
	\caption{$\hpi(\eta)^\top\underline 1$ in sets of parameters  $A_{1}$ (left) and $A_2$ (right). For set $A_1$ we have  $\eta_r =  1.47$ and $\eta_R =  0.6$, for set $A_2$ we get $\eta_r =  2.35$ and $\eta_R =  0.71$.}\label{fig:pieta_A}
\end{figure}
Figure \ref{fig:pi_1vs2_eta_A} shows the plots of the two-dimensional curves $(\hpi^1,\hpi^2)$ as functions of $\eta.$ Unlike Figure \ref{fig:pi_12_eta_A} we consider values of $\eta$ in the whole range $(0,\beta).$ We observe that for both parameter specifications, portfolio weights move in different directions on both sides of the hyperplane $\pi^1+\pi^2=1.$ Indeed, for set $A_1$ the optimal weight $\hat\pi^1$ decreases until $\eta$ reaches the threshold value $\eta_R,$ increases slightly until $\eta_r$ and then decreases  again, while $\hat\pi^2$ is always decreasing. For set $A_2$ optimal weight $\hat\pi^1$ is always decreasing, while $\hat\pi^2$ is always increasing. This can be explained mostly from the correlations with the continuous part of the insurable risk (positive for the first asset, negative for the second asset) and from the spreads of the mean return rate with the borrowing rate $\mu_i-R$, as they are positive for all cases except for the second asset in the parameter set $A_2.$ 
\begin{figure}[H]
	\centering
	\includegraphics[scale=0.4]{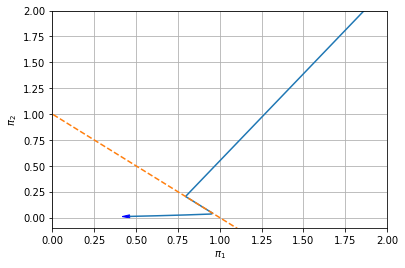}
	\includegraphics[scale=0.4]{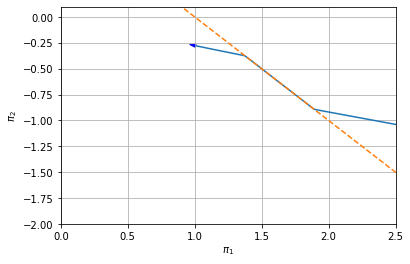}
	\caption{$\hpi_1(\eta) \ vs. \ \hpi_2(\eta)$ for parameter specifications $A_{1}$ (left) and $A_2$ (right).}\label{fig:pi_1vs2_eta_A}
\end{figure}
Table \ref{etaR-R} contains values of $\eta_R$ for the parameter set $A_2$ and different decreasing values of $R$, and figure \ref{plots_diff_etaR} contains the plots of $\hpi^1+\hpi^2$ and $\hpi^1$ vs. $\hpi^2$ as functions of $\eta$ for different values of $R.$ We note that the direction of $\hpi$ changes before reaching $\pi^1+\pi^2=1.$ 
\begin{table}
	\centering
	\begin{tabular}{|c|ccccccccc|}
		\hline	$R$ &  10\% & 9\% & 8\% & 7\% & 6\% & 5\% & 4\% & $3.5\%$ & $3.1\%$\\
		$\eta_R$ & $0.71$ & $0.95$  & $1.18$ & $1.42$ & $1.65$ & $1.89$ & $2.12$ & $2.24$ & $2.33$ \\\hline
	\end{tabular}
\caption{$\eta_R$ as function of borrowing rate $R>r.$}
\label{etaR-R}
\end{table}
From these, we see that the threshold value $\eta_R$ moves in opposite direction to the borrowing rate $R>r$, and $\eta_R\uparrow\eta_r$ as $R\downarrow r$. 
Similar behavior holds for $\eta_r$ as a function of the lending rate $r.$ However, note that the both $\eta_R$ and $\eta_r$ do not depend solely on the values $r$ and $R.$ 
\begin{figure}[H]
	\centering
	\includegraphics[scale=0.4]{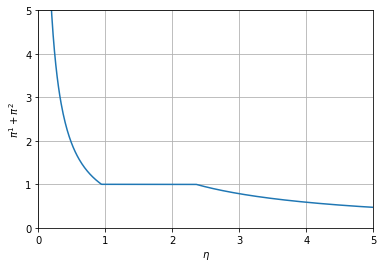}
	\includegraphics[scale=0.4]{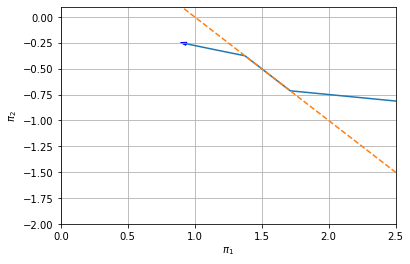}\\
	\includegraphics[scale=0.4]{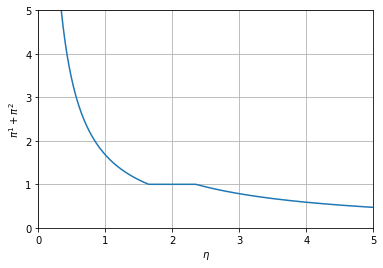}
	\includegraphics[scale=0.4]{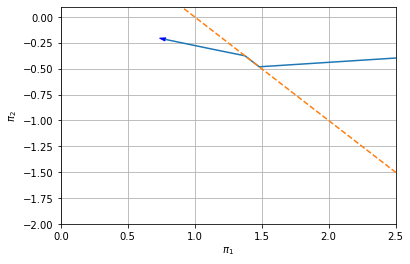}\\
	\includegraphics[scale=0.4]{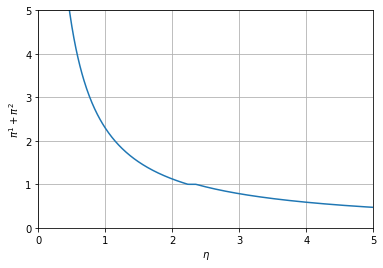}
	\includegraphics[scale=0.4]{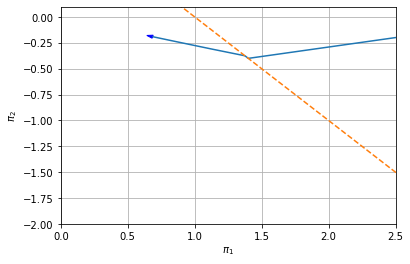}
	\caption{Plots of $\hpi^\top\underline 1$ (left) and $\hpi$ (right) as functions of $\eta>0$ for parameter set $A_2.$ From top to bottom we take $R = 9\%, 6\%, 3.5\%$}\label{plots_diff_etaR}
\end{figure}

In Figure \ref{fig:kappa_eta_A} we present plots of the optimal non-insured fraction $\hat\kappa$ of background risk, that is, the fraction assumed by the agent, as function of $\eta.$ Again, there exists a threshold value $\eta^*>0$ such that if $\eta\le\eta^*$ then the agent assumes totally the insurable background risk, whereas for $\eta>\eta^*$ the optimal fraction $\hat\kappa$ of insurable risk assumed by the agent is a decreasing convex function of $\eta.$
\begin{figure}[H]
	\centering
	\includegraphics[scale=0.4]{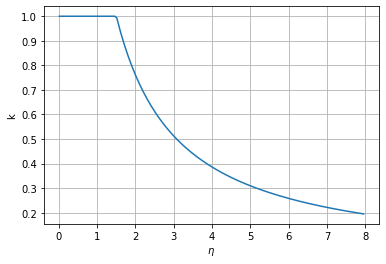}
	\includegraphics[scale=0.4]{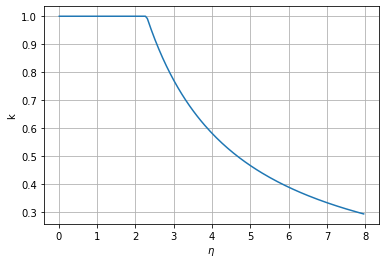}
	\caption{$\hat\kappa(\eta)$ for parameter specifications $A_{1}$ and $A_2.$}\label{fig:kappa_eta_A}
\end{figure}
Note that the range in which $\hat\kappa(\eta)$ remains constant equal to one is larger for set $A_2$ partially due to the cost of insurance $q$ being higher for this specification. 
\end{example}

\begin{example}
We now look into the relation of the optimal strategy with the correlation $\rho$ in  the case of only one risky asset ($d=1$). Again, we assume Beta distribution for the jump part of the insurable background risk, and consider four sets of parameters. For sets $B_1$ and $B_2$ we take $\alpha=2,$ $q=0.3,$ $\lambda =0.1$ and $\sigma=26\%$. For sets $C_1$ and $C_2$ we assume $\alpha=12$, $q=0.2,$  $\lambda =0.15$ and $\sigma=30\%$. For all four sets the borrowing and lending are the same: $r=3\%$, $R=9\%$. The remaining parameters are as follows
\begin{table}
	\centering
	\begin{tabular}{|c|c|c||c|c|}
		\hline
		Parameter &  $B_1$ &  $B_2$ &  $C_1$ &  $C_2$ \\
\hline
		$\beta$ & $8$ & $6$ & $8$ & $8$ \\
		$b$ & $0.4$ & $0.4$ & $0.8$ & $0.4$ \\
		$\mu$ & $16\%$ & $10 \%$  & $-5 \%$ & $16 \%$ \\
		$\eta$ & $2$ & $4$ & $4$ & $4$ \\
		\hline
	\end{tabular}
	\caption{Parameter sets $B_1,B_2,C_1$ and $C_2.$}\label{tab:parsetBC}
\end{table}
Figure \ref{fig:pirho} contains plots of the optimal proportion $\hpi$ of wealth invested in the risky asset as a function of the correlation $\rho$ with the continuous (diffusion) part of the insurable background risk. Note that in most cases the agent relies more in the risky asset as the correlation increases. Moreover, for set $B_1$ we have $\hpi(\rho)= 1$ in for $\rho \in [0.01,0.3]$, for set $B_2$ we get $\hpi(\rho)= 1$ for $\rho \in [0.6,0.77]$. In the notation of Corollary \ref{cor-diff-rates}, for these values of $\rho$ there exists $\xi^*\in[3\%,9\%]$ such that
\[
\frac{\mu-\xi^*}{\eta\sigma^2}+\frac{\rho b}{\sigma}\kappa(\xi^*;\rho)=1.
\] 
For set $C_1$ since $\frac{\mu-r}{\eta\sigma^2}=-0.2222<1$ and $\mu<r$, by part \emph{iv} of Corollary \ref{cor-diff-rates} $\hat\kappa=0$ and $\hpi=-0.2222$ are optimal for values of $\rho$ larger than
\begin{equation}\label{threshold-rho}
\frac{\sigma}{b(\mu-r)}[\lambda \Exp Y+p'(0)]=\frac{\sigma}{b(\mu-r)}\Bigl[\frac{\lambda\alpha}{\alpha+\beta}-q\Bigr]=0.5156.
\end{equation}
Similarly, for set $C_2$ we get $\hat\kappa=0$ and $\hpi=\frac{\mu-r}{\eta\sigma^2}= 0.3611$ are optimal in the region $\rho \in [-1,-0.6346]$.
\begin{figure}[H]
	\centering
	\includegraphics[scale=0.35]{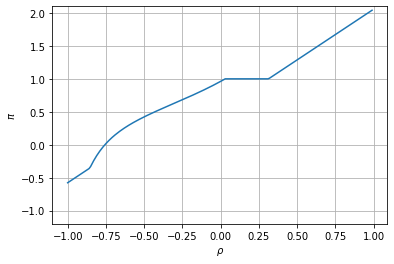}
	\includegraphics[scale=0.35]{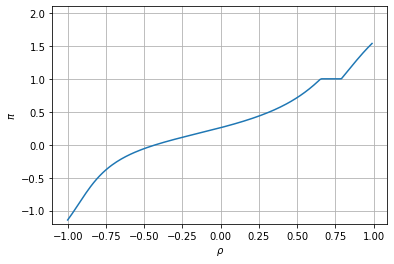}\\
	\includegraphics[scale=0.35]{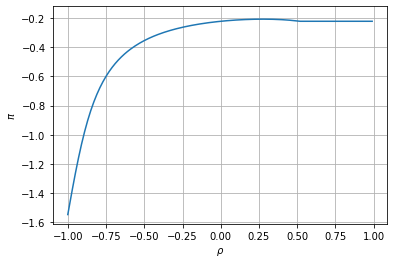}
	\includegraphics[scale=0.35]{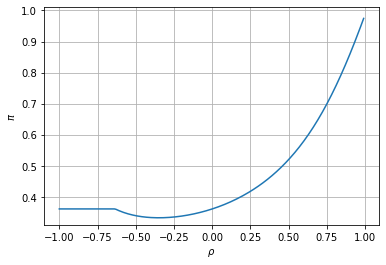}
	\caption{$\hpi(\rho)$ for sets $B_{1}, B_2$ (top) and $C_{1}, C_2$ (bottom).} \label{fig:pirho}
\end{figure}

Figure \ref{fig:kapparho} contains the plots of the optimal non-insured proportion $\hat\kappa$ with respect to the correlation $\rho$. The plots for sets $B_1$ and $B_2$ are U-shaped. In particular, for set $B_1$ if we further take $\rho=-1$ then
\begin{align*}
\lambda\Exp\Bigl[\frac{Y}{(1-Y)^\eta}\Bigr]-\frac{b\rho(\mu-r)}{\sigma}+\eta(1-\rho^2)-q&=-0.05714<0,\\
\frac{\mu-r}{\eta\sigma^2}+\frac{b\rho}{\sigma}&=-0,5769<1.
\end{align*}
By part \emph{vi} of Corollary \ref{cor-diff-rates} it is optimal for the agent to assume completely the insurable background risk, in this case by selling short the risky asset and benefiting from the negative (inverse) correlation. In fact, this choice is optimal near $\rho=-1.$ The same occurs for $\rho=1$ since
\begin{align*}
	\lambda\Exp\Bigl[\frac{Y}{(1-Y)^\eta}\Bigr]-\frac{b\rho(\mu-R)}{\sigma}+\eta(1-\rho^2)-q&=-0.3648<0,\\
	\frac{\mu-R}{\eta\sigma^2}+\frac{b\rho}{\sigma}&=2.0562>1.
\end{align*} 
Hence, by part \emph{vii} of Corollary \ref{cor-diff-rates} $\hat\kappa=1$ is again optimal, in this case by borrowing from the risk-free account to leverage a larger position in the risky asset and benefit from the positive (direct) correlation with the diffusion part of the background risk.
\begin{figure}[H]
	\centering
	\includegraphics[scale=0.35]{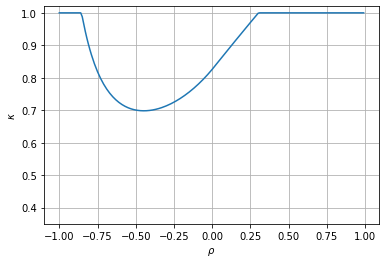}
	\includegraphics[scale=0.35]{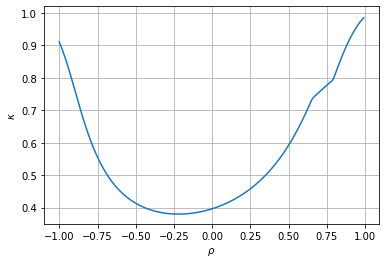}\\
	\includegraphics[scale=0.35]{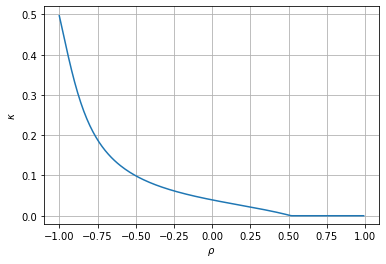}
	\includegraphics[scale=0.35]{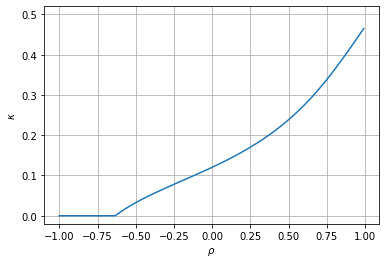}
	\caption{$\hat\kappa(\rho)$ for sets $B_{1}, B_2$ (top) and $C_{1}, C_2$ (bottom).} \label{fig:kapparho}
\end{figure}
For set $C_1$ we have that $\hat\kappa(\rho)$ is strictly decreasing for $\rho<0.5156$ and, as we already mentioned above, it stays constant equal to zero beyond that value. Since $\mu<r$ in this case this threshold value (\ref{threshold-rho}) for the correlation  increases with $q$, $\sigma$ and $\beta$, and decreases with $b$ and $\lambda.$

For parameter set $C_2$ we obtain the opposite behavior: $\hat\kappa(\rho)=0$ for $\rho<-0.6346$ and then increases strictly, but never reaches the value one. In general, if $\mu <r+\eta\sigma^2$ then the agent invests the Merton proportion in the risky asset and buys full insurance against the background risk as long as $\frac{\lambda\alpha}{\alpha+\beta}\geq q+\frac{b\rho}{\sigma}(\mu-r).$ Note this last inequality, which ensures optimality of $\hat\kappa=0,$ does not depend on the risk-aversion parameter $\eta.$ A similar assertion holds if $\mu>R+\eta\sigma^2.$

Finally, we look at the effect of first-order stochastic dominance of the jump distribution on optimal $\hat\kappa.$ We use the set of parameters $C_2$ and consider two values for $\alpha=2,12.$ In this case, Beta$(12,8)$ dominates Beta$(2,8),$ see Figure \ref{fig:beta_StDom}.
\begin{figure}[H]
	\centering
	\includegraphics[scale=0.35]{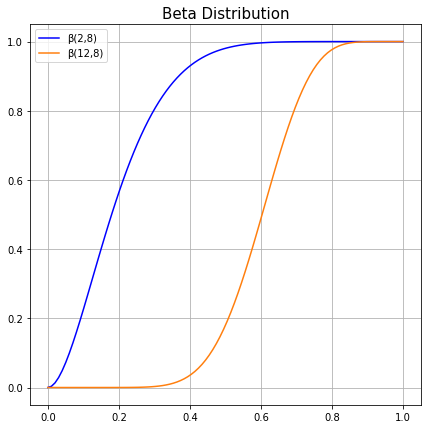}
	\ \ 
	\includegraphics[scale=0.35]{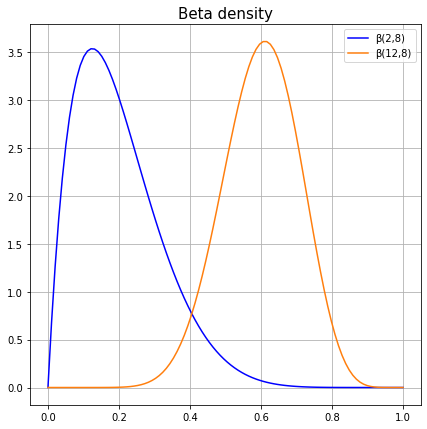}
	\caption{Cumulative distribution functions (left) and densities (right) of distributions Beta$(2,8)$ (blue) and Beta$(12,8)$ (orange) .} \label{fig:beta_StDom}
\end{figure}
Figure \ref{fig:kapparho_StDom} depicts $\hat{\kappa}(\rho)$ for these two jump distributions. As expected $\hat\kappa$ is larger for Beta$(2,8)$ and takes strictly positive values, unlike the case of Beta$(12,8)$ which remains zero for certain values of $\rho.$ Jumps with distribution Beta$(12,8)$ represent a greater risk than those with distribution Beta$(2,8)$.
\begin{figure}[H]
	\centering 
	\includegraphics[scale=0.6]{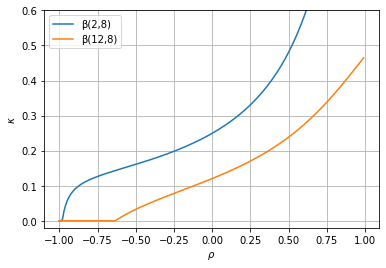}
	\caption{$\hat\kappa(\rho)$ for parameter set $C_2$ with jump distributions Beta$(12,8)$ (orange) and Beta$(2,8)$ (blue).} \label{fig:kapparho_StDom}
\end{figure}


\end{example}

\begin{remark}
	For the case $d=1,$ note that equation $h(\kappa,\xi)=0$ can be solved for $\xi$ explicitly for each $\kappa,$ so the zero-level curves can alternatively be characterized as the set of points $(\kappa,\bar\xi(\kappa))$ with
	\[
	\bar\xi(\kappa)=\mu-\frac{\sigma}{b\rho}\left[p'(\kappa)+(1-\rho^2)b^2\eta {\kappa}+\lambda\Exp\left(\frac{Y}{[1-{{\kappa}}Y]^\eta}\right)\right], \ \ k\in[0,1].
	\]
	Define
	\[
	\bar\pi(\kappa):=\pi(\bar\xi(\kappa))=\frac{1}{\eta b\sigma\rho}\left[p'(\kappa)+\lambda\Exp\left(\frac{Y}{[1-{{\kappa}}Y]^\eta}\right)\right]+\frac{b}{\rho\sigma}\kappa.
	\]
	Then
	\[
	\bar\pi'(\kappa)=\frac{1}{\eta b\sigma\rho}\left[p''(\kappa)+\eta\lambda\Exp\left(\frac{Y^2}{[1-{{\kappa}}Y]^{1+\eta}}\right)\right]+\frac{b}{\rho\sigma}.
	\]
	That is, $\bar\pi$ is strictly increasing (resp. strictly decreasing) if $\rho>0$ (resp. if $\rho<0$). If $\rho>0$ and
	\[
	p'(0)+\lambda\Exp Y<b\sigma\eta\rho<p'(1)+\lambda\Exp\left(\frac{Y}{[1-Y]^\eta}\right)+\eta b^2
	\]
	then there exists $\kappa^*$ such that $\bar\pi(\kappa^*)=1,$ and the optimal strategy can be rewritten as
	\[
	(\hpi,\hat\kappa)
	=\left\{
	\begin{array}{ll}
		(\pi(r),\kappa(r))=(\bar\pi(\kappa(r)),\kappa(r)), & \ \bar\xi(\kappa^*)<r\\\\
		(\pi(R),\kappa(R))=(\bar\pi(\kappa(R)),\kappa(R)), & \ \bar\xi(\kappa^*)>R\\\\
		(1,\kappa^*), & \ r\le \bar\xi(\kappa^*)\le R
	\end{array}
	\right.
	\]
	A similar result holds for the case $\rho<0.$
\end{remark}
We observe in Figure \ref{fig:pi_1vs2_eta_A} that optimal portfolios move along 
one-dimensional portfolio lines. Indeed, the following mutual-fund separation-type result holds for the case of linear premium function. Let $\pi(\xi,\eta)$  and $\kappa(\xi,\eta)$ be as in Corollary \ref{cor-diff-rates}, with risk aversion level $\eta>0$ as an additional variable.
\begin{theorem}[Mutual-Fund Separation Theorem]\label{mutual}
Suppose the premium function is linear $p(\kappa)=q(1-\kappa)$ with premium rate $q>0$. Let $\eta_1<\bar\eta<\eta_2.$ 
\begin{enumerate}
	\item If $\pi(r,\eta_1)^\top\underline 1<1$ and $\pi(r,\eta_2)^\top\underline 1<1$ then there exists $\delta\in(0,1)$ such that
	\[
	\delta(\pi(r,\eta_1),\kappa(r,\eta_1))+(1-\delta)(\pi(r,\eta_2),\kappa(r,\eta_2))
	\]
	is optimal for the risk aversion level $\bar\eta.$
	
	\item  $\pi(R,\eta_1)^\top\underline 1>1$ and $\pi(R,\eta_2)^\top\underline 1>1$ then there exists $\delta\in(0,1)$ such that
	\[
	\delta(\pi(R,\eta_1),\kappa(R,\eta_1))+(1-\delta)(\pi(R,\eta_2),\kappa(R,\eta_2))
	\]
	is optimal for the risk aversion level $\bar\eta.$
	
	\item  $\pi(\xi_1,\eta_1)^\top\underline 1=\pi(\xi_2,\eta_2)^\top\underline 1=1$ for some $\xi_1,\xi_2\in[r,R]$ then there exists $\delta\in(0,1)$ such that
	\[
	\delta(\pi(\xi_1,\eta_1),\kappa(\xi_1,\eta_1))+(1-\delta)(\pi(\xi_2,\eta_2),\kappa(\xi_2,\eta_2))
	\]
	is optimal for the risk aversion level $\bar\eta.$
\end{enumerate}
\end{theorem}
\begin{proof}
Recall that in the case of different rates for borrowing and lending, and linear premium function, we have $\tilde f((\xi-r)\underline 1,\gamma)=\xi-r-q$ for $\xi\in[r,R]$ and $\gamma\le -q.$ The associated optimality condition  (\ref{optcond-zeta-pikappa}) becomes
\[
-(R-r)(\pi^\top\underline 1 -1)^+ +\pi^\top \nabla_\pi H(\pi,\kappa;\eta)+\kappa [\partial_\pi H(\pi,\kappa;\eta)+q]=\xi-r
\]
and $\nabla_\pi H(\pi,\kappa;\eta)=(\xi-r)\underline 1$ for some $\xi\in[r,R].$ In view of this, for case \emph{1} we define
$L:(0,\infty)\times[0,1]\times[r,R]\to\R$ as
\begin{align*}
L(\delta,\eta;\xi):= &\,\tilde\pi(\delta;\xi)^\top[\nabla_\pi H(\tilde\pi(\delta;\xi),\tilde\kappa(\delta;\xi);\eta)-(\xi-r)\underline 1]\\
&+\tilde\kappa(\delta;\xi)[\partial_\kappa  H(\tilde\pi(\delta;\xi),\tilde\kappa(\delta;\xi);\eta)+q]
	\end{align*}
where for each $\delta\in[0,1]$ and $\xi\in[r,R]$ we denote
\begin{align*}
	\tilde \pi(\delta;\xi)&=\delta\pi(\xi,\eta_1)+(1-\delta)\pi(\xi,\eta_2),\\
	\tilde\kappa(\delta;\xi)&=\delta\kappa(\xi,\eta_1)+(1-\delta)\kappa(\xi,\eta_2).
\end{align*}
By definition of $\pi(\xi,\eta)$  and $\kappa(\xi,\eta),$ we have
\begin{align*}
\nabla_\pi H(\tilde\pi(0;r),\tilde\kappa(0;r);\eta_2)=\nabla_\pi H(\tilde\pi(1;r),\tilde\kappa(1;r);\eta_1)&=\underline 0,\\
\partial_\kappa H(\tilde\pi(0;r),\tilde\kappa(0;r);\eta_2)=\partial_\kappa H(\tilde\pi(1;r),\tilde\kappa(1;r);\eta_1)&=-q.
\end{align*}
Using the expressions (\ref{grad_partialH}) it is easy to see that $L$ is strictly decreasing in $\eta.$ Then, taking $\xi=r$ for case \emph{1} we get
\[
L(1,\bar\eta;r)<L(1,\eta_1;r)=0=L(0,\eta_2;r)=0<L(0,\bar\eta;r).
\]
By the intermediate value Theorem, there exists $\delta\in(0,1)$ such that $L(\delta,\bar\eta;r)=0.$ From the same argument in the proof of part \emph{i}. in Corollary \ref{cor-diff-rates} and Lemma \ref{lemma-optcond}, we get that  $(\tilde\pi(\delta;r),\tilde\kappa(\delta;r))$ is optimal for the risk-aversion level $\bar\eta.$ The argument is the same for \emph{2} taking $\xi=R,$ as in this case we have
\[
\nabla_\pi H(\tilde\pi(0;R),\tilde\kappa(0;R);\eta_2)=\nabla_\pi H(\tilde\pi(1;R),\tilde\kappa(1;R);\eta_1)=(R-r)\underline 1.
\]
For case \emph{3}, we suppose $\pi(\xi_1,\eta_1)^\top\underline 1=\pi(\xi_2,\eta_2)^\top\underline 1=1$ for $r\le\xi_1<\xi_2\le R.$ We now denote
\begin{align*}
	\tilde\pi(\delta)&=\delta\pi(\xi_1,\eta_1)+(1-\delta)\pi(\xi_2,\eta_2),\\
	\tilde\kappa(\delta)&=\delta\kappa(\xi_1,\eta_1)+(1-\delta)\kappa(\xi_2,\eta_2),
\end{align*}
and
\[
L(\delta,\eta,\xi):=\tilde\pi(\delta)^\top[\nabla_\pi H(\tilde\pi(\delta),\tilde\kappa(\delta);\eta)-(\xi-r)\underline 1]+\tilde\kappa(\delta)[\partial_\kappa  H(\tilde\pi(\delta),\tilde\kappa(\delta);\eta)+q]
\]
Let $\xi^*\in (\xi_1,\xi_2)$ be fixed. Since  $L$ is strictly decreasing in both $\eta$ and $\xi,$ and
\begin{align*}
	\nabla_\pi H(\tilde\pi(0),\tilde\kappa(0);\eta_2)&=(\xi_2-r)\underline 1,\\
	\nabla_\pi H(\tilde\pi(1),\tilde\kappa(1);\eta_1)&=(\xi_1-r)\underline 1,
\end{align*}
we obtain
\[
L(1,\bar\eta,\xi^*)< L(1,\eta_1,\xi_1)=0=L(0,\eta_2,\xi_2)<L(0,\bar\eta,\xi^*).
\]
Again, by the intermediate value Theorem, there exists $\delta\in(0,1)$ such that $L(\delta,\bar\eta,\xi^*)=0.$ By Lemma \ref{lemma-optcond} we conclude that the pair $(\tilde\pi(\delta),\tilde\kappa(\delta))$ is optimal for the risk-aversion level $\bar\eta.$
\end{proof}

\subsubsection{Large-investor with piecewise constant price pressure}
Let $d=1$ and $g(\pi)=\pi l(\pi)$ with $l=m^+ \mathbf{1}_{[0,+\infty)}+m^- \mathbf{1}_{(-\infty,0)},$ see e.g. Cuoco and Cvitani\'{c} \cite{cuoco1998} and Long \cite[Section 7.1]{long2004investment}. We assume $m^+<m^-$ so that $f$ is concave. In particular, if $m^+<0<m^-,$ long positions in the risky asset put pressure on its expected return, while
short positions increase its expected return.

In this case it is easy to see that $\tilde g=0$ with effective domain $\N=[-m^-,-m^+]\times\R.$ The same arguments of
Section \ref{exdiffrates-2} lead to the following result
\begin{corollary}\label{large}
	Suppose for each $m\in[m^+,m^-]$ there exists $\kappa(m)\in[0,1]$ solution to $h(\kappa;m)=0$ with
\[
h(\kappa;m):=\frac{b\rho}{\sigma}({\mu}+m-r)-\eta b^2(1-\rho^2){\kappa}-\lambda\Exp\Bigl[\frac{Y}{(1-\kappa Y)^\eta}\Bigr]-p'(\kappa)
\]
and set
	\[
	\pi(m):=\frac{{\mu}+m-r}{\eta\sigma^2}+\frac{1}{\sigma}\rho b \kappa(m).
	\]
	\begin{enumerate}
		\item[i.] If $\pi^+:=\pi(m^+)>0$ then $(\pi^+,\kappa(m^+))$ is optimal. In this case, the agent holds a long position in the risky asset.
		\item[ii.] If $\pi^-:=\pi(m^-)<0$ then $(\pi^-,\kappa(m^-))$ is optimal. In his case the agent holds a short position in the risky asset.
		
		\item[iii.] If there exists $m\in[m^+,m^-]$ such that $\pi(m)=0$ 	then $(0,\kappa(m))$ is optimal. In particular, the agent invests all in the risk-free asset. 
		
	\item[iv.] Suppose $\lambda \Exp[Y]+p'(0)\geq b\rho(\mu-r+m^+)/\sigma$ and $\mu+m^+>r.$ Then the pair $\hpi=(\mu-r+m^+)/\eta\sigma^2,$ $\hat\kappa=0$ is optimal. In particular, the agent holds a long position in the risky asset and insures totally against the adverse shocks.  

\item[v.] Suppose $\lambda \Exp[Y]+p'(0)\geq b\rho(\mu-r+m^-)/\sigma$ and $\mu+m^-<r.$ Then the pair $\hpi=(\mu-r+m^-)/\eta\sigma^2,$ $\hat\kappa=0$ is optimal. In particular, the agent holds a short position in the risky asset and insures totally against the adverse shocks.

\item[vi.] Suppose $\lambda\Exp\Bigl[\frac{Y}{(1-Y)^\eta}\Bigr]+\eta b^2(1-{\rho}^2)+p'(1)<b\rho(\mu-r+m^+)/\sigma$ and
\[
\mu+m^++\eta\sigma\rho b\geq r.
\]
Then the pair $\hpi=\frac{\mu-r+m^+}{\eta\sigma^2}+\frac{b\rho}{\sigma},$ $\hat\kappa=1$ is optimal. In this case the agent assumes totally the insurable background risk and holds a long position in the risky asset.  
\item[vii.] Suppose $\lambda\Exp\Bigl[\frac{Y}{(1-Y)^\eta}\Bigr]+\eta b^2(1-{\rho}^2)+p'(1)<b\rho(\mu-r+m^-)/\sigma$  and
\[
\mu+m^-+\eta\sigma\rho b\leq r.
\]
Then the pair $\hpi=\frac{\mu-r+m^-}{\eta\sigma^2}+\frac{b\rho}{\sigma},$ $\hat\kappa=1$ is optimal. In this case the agent assumes totally the insurable background risk and holds a short position in the risky asset.  
	\end{enumerate}
\end{corollary}

\section{Linear premium function with premium rate depending on risky asset allocation}

As a final example, we consider the case $d=1$ and $f(\pi,\kappa)=-(1-\kappa)q(\pi)$ with $q$ a positive-valued differentiable convex function of the portfolio proportion $\pi\in\R.$ Such premium functions can be used to model insurers that incorporate information about the customer in the premium rate, thus offering coverage at a cost that reflects customers' actual risk exposure and preventing asymmetric information and adverse selection from being exploited. Clearly, this requires that the agent reveals all information on portfolio holdings to the potential insurer.

Although this assumption represents mostly a hypothetical scenario, our model and the following result shed some light on how the agent can make an optimal decision on both investment and insurance demand, given the correlation between the associated risks. Moreover, under the SEC\footnote{Securities and Exchange Comission} regulation, quarterly disclosure of portfolio holdings in public companies by institutional investment fund managers is mandatory, so the assumption is actually not too far from the reality of regulated investment companies such as mutual and  exchange-traded funds.

\begin{theorem}
Suppose  $q'$ is bounded and  sufficiently small so that
\begin{equation}\label{SOC}
[q'(\pi)+\eta \rho b\sigma]^2<\sigma^2\eta^2\left(b^2+\lambda\Exp[Y^2]\right), \ \ \forall\pi\in\R.
\end{equation}
Let
\begin{align*}
	Q(\pi)&:=q(\pi)+\eta b\sigma\rho\pi,\\
	G(\kappa)&:=\eta b^2\kappa+\lambda\Exp\Bigl[\frac{Y}{(1-\kappa Y)^\eta}\Bigr].
\end{align*}
Suppose further $Q$ is invertible on $\range G=\Bigl[\lambda\Exp Y,\eta b^2+\lambda\Exp\Bigl(\frac{Y}{[1-Y]^\eta}\Bigr)\Bigr]$ and there exists $\hat\kappa\in(0,1)$ such that $G(\hat\kappa)\in \range Q$ and
\begin{equation}\label{FOC}
\mu-r-\eta\sigma^2Q^{-1}(G(\hat\kappa))+\eta\sigma\rho b\hat\kappa-q'(Q^{-1}(G(\hat\kappa)))(1-\hat\kappa)=0
\end{equation}
then $\hat\pi:=Q^{-1}(G(\hat\kappa))$ and $\hat\kappa$ are optimal.
\end{theorem}
\begin{proof}
The necessary first-order optimality conditions for existence of an interior optimal are
\begin{align*}
	q(\pi)+\eta b[\rho\sigma\pi -b\kappa]-\lambda\Exp\Bigl[\frac{Y}{(1-\kappa Y)^\eta}\Bigr]&=0,\\
	-q'(\pi)(1-\kappa)+\mu-r-\eta\sigma[\sigma\pi-\rho b\kappa]&=0,
\end{align*}
which are satisfied if (\ref{FOC}) holds. These in turn become sufficient if the second-order condition
\[
\eta\left[(1-\kappa)q''(\pi)+\eta\sigma^2\right]\left\{b^2+\lambda\Exp\Bigl[\frac{Y^2}{(1-\kappa Y)^{1+\eta}}\Bigr]\right\}>[q'(\pi)+\eta \rho b\sigma]^2
\]
is satisfied, which occurs if (\ref{SOC}) holds.
\end{proof}
\begin{remark}\label{rem-qprime}
If $\rho>0$ it suffices that $q'(\pi)>-\eta \rho b\sigma$ for all $\pi\in\R$ for $Q(\pi)$ to be strictly increasing and invertible on $\range G.$ The same assertion holds if $\rho<0$ and $q'(\pi)<-\eta \rho b\sigma$ for all $\pi\in\R.$
\end{remark}

\begin{example}
We set $q(\pi) := \lambda \mathbb{E}[Y] +C(\sqrt{ \pi^2+ A^2}-A).$ Note that $q(0)=\lambda \Exp Y,$ that is the fair premium if the agent has no exposure to the financial market.  The derivative $q'(\pi) = \frac{\pi C}{\sqrt{ \pi^2+ A^2}}$ satisfies $q'(\pi)>-\eta \rho b\sigma$ (resp. $q'(\pi)<-\eta \rho b\sigma$) if $\rho>0$ (resp. if $\rho<0$) and $C>\eta b \rho\sigma$ (resp. $C<-\eta b \rho\sigma$). Moreover, $\abs{q'(\pi)}<C$ so (\ref{SOC}) holds, for instance, if
\begin{equation}\label{SOC-2}
2[C^2+(\eta\rho b\sigma)^2]\le\eta^2\sigma^2\left\{b^2+\lambda\Exp[ Y^2]\right\}.
\end{equation}
As before, we assume $Y\sim \textrm{Beta}(\alpha,\beta).$ Then, (\ref{SOC-2}) reads
\begin{equation}
2[C^2+(\eta\rho b\sigma)^2]\le\eta^2\sigma^2\left\{b^2+\frac{\lambda\alpha(\alpha+1)}{(\alpha+\beta)(\alpha+\beta+1)}\right\}.
\end{equation}
If $\rho>0,$ the last inequality together with $C>\eta b \rho\sigma$ imply that the following must hold for $\eta$ 
\begin{equation}\label{interval-C-eta}
\frac{(b \rho\sigma)^2}{C^2}\le \frac{1}{\eta^2}\le \frac{\sigma^2}{2C^2}\left\{b^2(1-2\rho^2)+\frac{\lambda\alpha(\alpha+1)}{(\alpha+\beta)(\alpha+\beta+1)}\right\}.
\end{equation}
This condition implicitly requires that the following holds 
\[
b^2(1-4\rho^2)+\frac{\lambda\alpha(\alpha+1)}{(\alpha+\beta)(\alpha+\beta+1)}\geq 0.
\]
Note that this always holds if $0\le\rho\le 1/2.$ If, instead $\rho<0$ we can similarly deduce that $\eta$ must satisfy the following
\[
\frac{1}{\eta^2}\le \frac{\sigma^2}{C^2}\min\left\{\frac{1}{2}\left[b^2(1-2\rho^2)+\frac{\lambda\alpha(\alpha+1)}{(\alpha+\beta)(\alpha+\beta+1)}\right],(b\rho)^2\right\}.
\] 
Figure \ref{fig:StDom_s5} shows plots of $\hat{\kappa}$ as a function of $\rho$ for jumps with distributions Beta$(2,8)$ (in blue) and Beta$(12,8)$ (in orange), as the same parameter set $C_2$ from previous sections. We observe that monotonicity with respect to stochastic dominance also holds for this example. 
\begin{figure}[H]
	\centering
	\includegraphics[scale=0.5]{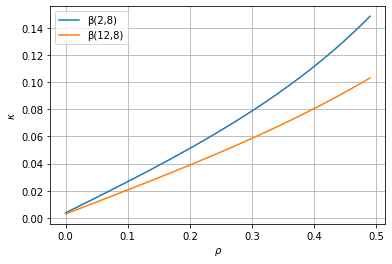}
	\caption{$\hat{\kappa}(\rho)$ for $f(\pi, \kappa) = -(1-\kappa)q(\pi)$ and parameter set $C_2$ with jump distributions Beta$(12, 8)$ (orange) and
Beta$(2, 8)$ (blue)}\label{fig:StDom_s5}
\end{figure}

\end{example}

\section{Conclusions}

This article considers a hypothetical setting that extends the findings of \cite{touzi2000optimal} and \cite{lin2012risky} to the case in which the agent's decision on portfolio choice and insurance demand cause nonlinear frictions in the dynamics of the wealth process. Using the associated HJB PDE, we found conditions under which the agent assumes the insurable risk entirely, partially, or purchases total insurance against it. We also proved that the insurance demand increases with the first-order stochastic dominance of the jump size and the magnitude of the jump arrival rate. 

We have paid special attention to piece-wise linear portfolio allocation frictions, with funding costs arising from differential borrowing and lending rates being the most emblematic example. We found sufficient conditions  under which the agent responds to the insurable background risk by investing at the (lower) lending risk-free rate $r$, holding only investments in risky assets, or leveraging the risky-asset portfolio by borrowing at the (higher) funding rate $R.$ The optimal investment allocation is given by the Merton proportion plus a hedging component against the insurable exogenous jump risk, so the agent uses the optimal demand in risky assets to manage its exposure to both financial and background risk. The correlations also contribute to the diversification effect. Finally, we proved a mutual-fund separation result which shows that optimal portfolio allocations move along one-dimensional segments. Hence, the optimal allocation for a given risk tolerance level can be obtained as a combination of two mutual funds. 

\subsubsection*{Acknowledgments}

This work was supported by Alianza EFI-Colombia Cientifica grant codes 60185 and FP44842-220-2018.

\section*{Declarations}

\begin{itemize}
\item Conflict of interest/Competing interests: Both authors certify that they have no affiliations with or involvement in any organization or entity with any financial interest or non-financial interest in the subject matter or materials discussed in this manuscript.

\item Authors' contributions: Both authors contributed equally to this manuscript.
\end{itemize}

\bibliographystyle{plainnat}
\bibliography{biblio_insurable_risk}

\end{document}